\documentclass[a4paper,12pt,eqno]{article}

\usepackage{amsthm,amsmath,amssymb}
\usepackage{mathrsfs}
\usepackage{wrapfig}

\usepackage{graphicx}

\usepackage{color}

\theoremstyle{plain}
\newtheorem{theorem}{Theorem}[section]
\newtheorem{lemma}[theorem]{Lemma}

\newtheorem{proposition}[theorem]{Proposition}

\theoremstyle{definition}

\newtheorem{example}[theorem]{Example}

\theoremstyle{remark}
\newtheorem{remark}[theorem]{Remark}

\makeatletter
 \@addtoreset{equation}{section}
\makeatother

\newcommand{\RM}{\mathbb{R}}

\newcommand{\CM}{\mathbb{C}}

\newcommand{\HM}{\mathbb{H}}
\newcommand{\Mat}{\operatorname{Mat}}

\newcommand{\Spec}{\operatorname{Spec}}

\title{\bf The discrete-time quaternionic quantum walk and the second weighted zeta function on a graph}

\author{
{\small Norio Konno}\\
{\scriptsize Department of Applied Mathematics, 
Faculty of Engineering, 
Yokohama National University}\\
{\scriptsize Hodogaya, Yokohama 240-8501, Japan}\\
{\scriptsize e-mail: konno@ynu.ac.jp}\\
{\small Hideo Mitsuhashi}\\
{\scriptsize Faculty of Education, 
Utsunomiya University}\\
{\scriptsize Utsunomiya, Tochigi 321-8505, Japan}\\
{\scriptsize e-mail: mitsu@cc.utsunomiya-u.ac.jp}\\
{\small Iwao Sato}\\
{\scriptsize Oyama National College of Technology}\\
{\scriptsize Oyama, Tochigi 323-0806, Japan}\\
{\scriptsize e-mail: isato@oyama-ct.ac.jp}\\}

\date{\small Mathematics Subject Classifications: 60F05, 05C50, 15A15, 11R52}

\begin{document}

\maketitle

\begin{abstract}
We define the quaternionic quantum walk on a finite graph and investigate its properties. 
This walk can be considered as a natural quaternionic extension of the Grover walk on a graph. 
We explain the way to obtain all the right eigenvalues of a quaternionic matrix and a 
notable property derived from the unitarity condition for the quaternionic quantum walk. 
Our main results determine all the right eigenvalues of the quaternionic quantum
walk by using complex eigenvalues of the quaternionic weighted matrix 
which is easily derivable from the walk.  
Since our derivation is owing to a quaternionic generalization of the determinant expression of 
the second weighted zeta function, 
we explain the second weighted zeta function and 
the relationship between the walk and the second weighted zeta function. 

\end{abstract}

\setcounter{equation}{0}
\section{Introduction}

The discrete-time quaternionic quantum walk on a graph is a quantum process on a graph 
which is governed by a unitary matrix. 
The study of quantum walks started in earnest as quantum versions of random walks 
around the end of the last century, and quantum walks have been developed rapidly for 
more than two decades in connection with various fields such as quantum information science 
and quantum physics. 
Detailed information on quantum walks can be found in several books at present, for example, 
Manouchehri and Wang \cite{MW2013}, Portugal \cite{Port2013}, 
Konno \cite{Konno2014}. 
An important example of the quantum walk on a graph 
is the Grover walk which originates from Grover's algorithm. 
Grover's algorithm which was introduced in 
\cite{Grover1996} is a quantum search algorithm that performs quadratically 
faster than the best classical search algorithm. 
Later various researchers investigated the Grover walk and developed 
the theory of discrete-time quantum walks intensively.

Recently, Konno \cite{Konno2015} established a quaternionic extension of quantum walks. 
These are a quaternionic extension of quantum walks and 
can be viewed as quaternionic quantum dynamics. 
One of the important backgrounds of quaternionic quantum walk is quaternionic quantum mechanics. 
The origin of quaternionic quantum mechanics goes back to 
the axiomatization of quantum mechanics by Birkhoff and von Neumann in 1930s. 
After that, the subject was studied further by Finkelstein, Jauch, and Speiser, 
and more recently by Adler and others. 
One significant motivation of studying quaternionic quantum mechanics is that 
physical reality might be described by quaternionic quantum system at the fundamental level, 
and this dynamics is described asymptotically by the (ordinary) quantum field theory 
at the level of all presently known physical phenomena. 
A detailed exposition of quaternionic quantum mechanics can be found in Adler \cite{Adler1995}.

On the other hand, Zeta functions of graphs have been investigated for half a century. 
Their origin is the Ihara zeta function which was defined by Ihara \cite{Ihara1966}, 
and various extensions have appeared so far. 
Among them we focus on the second weighted zeta function of a graph. 
The second weighted zeta function which was proposed by 
Sato \cite{Sato2007} is a multi-weighted version of the Ihara zeta function, and 
has several applications in discrete-time quantum walks and quantum graphs. 
For example, the second weighted zeta function 
played essential roles in the concise proof of the spectral mapping theorem 
for the Grover walk on a graph in \cite{KS2012}.

In this paper, we define the discrete-time quaternionic quantum walk on a graph as 
a quaternionic extension of the Grover walk on a finite graph, and discuss its 
right spectrum and the relationship between the walk and the second weighted zeta 
function of a graph. Our results can be viewed as a generalization of \cite{KMS2016}.


\section{The Grover walk on a graph}

Let $G=(V(G)$, $E(G))$ be a finite connected graph with the set $V(G)$ of 
vertices and the set $E(G)$ of undirected edges $uv$ 
joining two vertices $u$ and $v$. 
We assume that $G$ is finite connected and has neither loops nor multiple edges throughout. 
For $uv \in E(G)$, we mean by an arc $(u,v)$ the directed edge from $u$ to $v$. 
Let $D(G)=\{\,(u,v),\,(v,u)\,\mid\,uv{\;\in\;}E(G)\}$ and 
$|V(G)|=n,\;|E(G)|=m,\;|D(G)|=2m$. 
For $e=(u,v){\;\in\;}D(G)$, $o(e)=u$ denotes the {\it origin} and $t(e)=v$ the {\it terminal} 
of $e$ respectively. 
Furthermore, let $e^{-1}=(v,u)$ be the {\em inverse} of $e=(u,v)$. 
The {\em degree} $d_u=\deg u = \deg {}_G \  u$ of a vertex $u$ of $G$ is the number of edges 
incident to $u$. 
A {\em path $P$ of length $\ell$} in $G$ is a sequence 
$P=(e_1, \cdots ,e_{\ell})$ of $\ell$ arcs such that $e_i \in D(G)$ and 
$t(e_i)=o(e_{i+1})$ for $i{\;\in\;}\{1,\cdots,\ell-1\}$. 
We set $o(P)=o(e_1)$ and $t(P)=t(e_{\ell})$. $|P|$ denotes the length of $P$. 
A path $P=(e_1, \cdots ,e_{\ell})$ is said to be a {\em cycle} if $t(P)=o(P)$ and 
to have a {\em backtracking} 
if $ e_{i+1} =e_i^{-1} $ for some $i(1{\leq}i{\leq}\ell-1)$. 
The {\em inverse} of a path $P=(e_1,\cdots,e_{\ell})$ is the path 
$(e_{\ell}^{-1},\cdots,e_1^{-1})$ and is denoted by $P^{-1}$. 

We give a definition of the discrete-time quantum walk on $G$. 
Let $\mathscr{H}=\oplus_{e{\in}D(G)}{\CM}|e\rangle$ be the finite dimensional Hilbert 
space spanned by arcs of $G$. 
The transition matrix ${\bf U}$ of a discrete-time quantum walk consists of 
the following two consecutive operations: 
\begin{enumerate}
\renewcommand{\labelenumi}{\rm \arabic{enumi}.}
\item
For each $u{\;\in\;}V$, we perform a unitary transformation ${\bf C}_u$ 
on the states $|f\rangle$ that satisfy $t(f)=u$. 
\item
For all $e{\;\in\;}D(G)$, we perform the shift ${\bf S}$ that is defined by 
${\bf S}|e{\rangle}=|e^{-1}{\rangle}$. 
\end{enumerate} 
The transition matrix ${\bf U}^{\rm Gro}$ of the Grover walk on $G$ is defined by setting 
the Grover's diffusion matrix as ${\bf C}_u$: 
\begin{equation*}
{\bf C}_u=\begin{pmatrix}
-1+\frac{2}{d_u} & \frac{2}{d_u} & \cdots & \frac{2}{d_u} \\
\frac{2}{d_u} & -1+\frac{2}{d_u} & \cdots & \frac{2}{d_u} \\
\cdots & \cdots & \cdots & \cdots \\
\frac{2}{d_u} & \frac{2}{d_u} & \cdots & -1+\frac{2}{d_u}
\end{pmatrix},
\end{equation*}
Then ${\bf U}^{\rm Gro}=({\bf U}^{\rm Gro}_{ef})_{e,f \in D(G)}$ is given by 
\[
{\bf U}^{\rm Gro}_{ef} =\left\{
\begin{array}{ll}
2/d_{o(e)} (=2/d_{t(f)} ) & \mbox{if $t(f)=o(e)$ and $f \neq e^{-1} $, } \\
2/d_{o(e)} -1 & \mbox{if $f= e^{-1} $, } \\
0 & \mbox{otherwise.}
\end{array}
\right. 
\]
${\bf U}^{\rm Gro}$ is called the {\it Grover matrix}.

We denote by ${\Spec}({\bf A})$ the multiset of eigenvalues of a complex square matrix 
${\bf A}$ counted with multiplicity. 
We shall give examples of Grover walks and their spectra. 

\begin{example}
$\quad G=K_3$. Then $d_{o(e)}=2$ for all $e{\in}D(G)$ and ${\bf U}^{\rm Gro}$ is given as follows: 
\[
\hspace{4cm}
{\bf U}^{\rm Gro}= 
\bordermatrix{
\!\! & e_1 & e_1^{-1} & e_2 & e_2^{-1} & e_3 & e_3^{-1} \cr
e_1\!\! & 0 & 0 & 0 & 0 & 1 & 0 \cr
e_1^{-1}\!\! & 0 & 0 & 0 & 1 & 0 & 0 \cr
e_2\!\! & 1 & 0 & 0 & 0 & 0 & 0 \cr
e_2^{-1}\!\! & 0 & 0 & 0 & 0 & 0& 1 \cr
e_3\!\! & 0 & 0 & 1 & 0 & 0 & 0 \cr
e_3^{-1}\!\! & 0 & 1 & 0 & 0 & 0 & 0 \cr
}.
\]

\vspace{-4.0cm}
\begin{figure}[htbp] 
\hspace{2cm}
  \includegraphics[width=4.0cm]{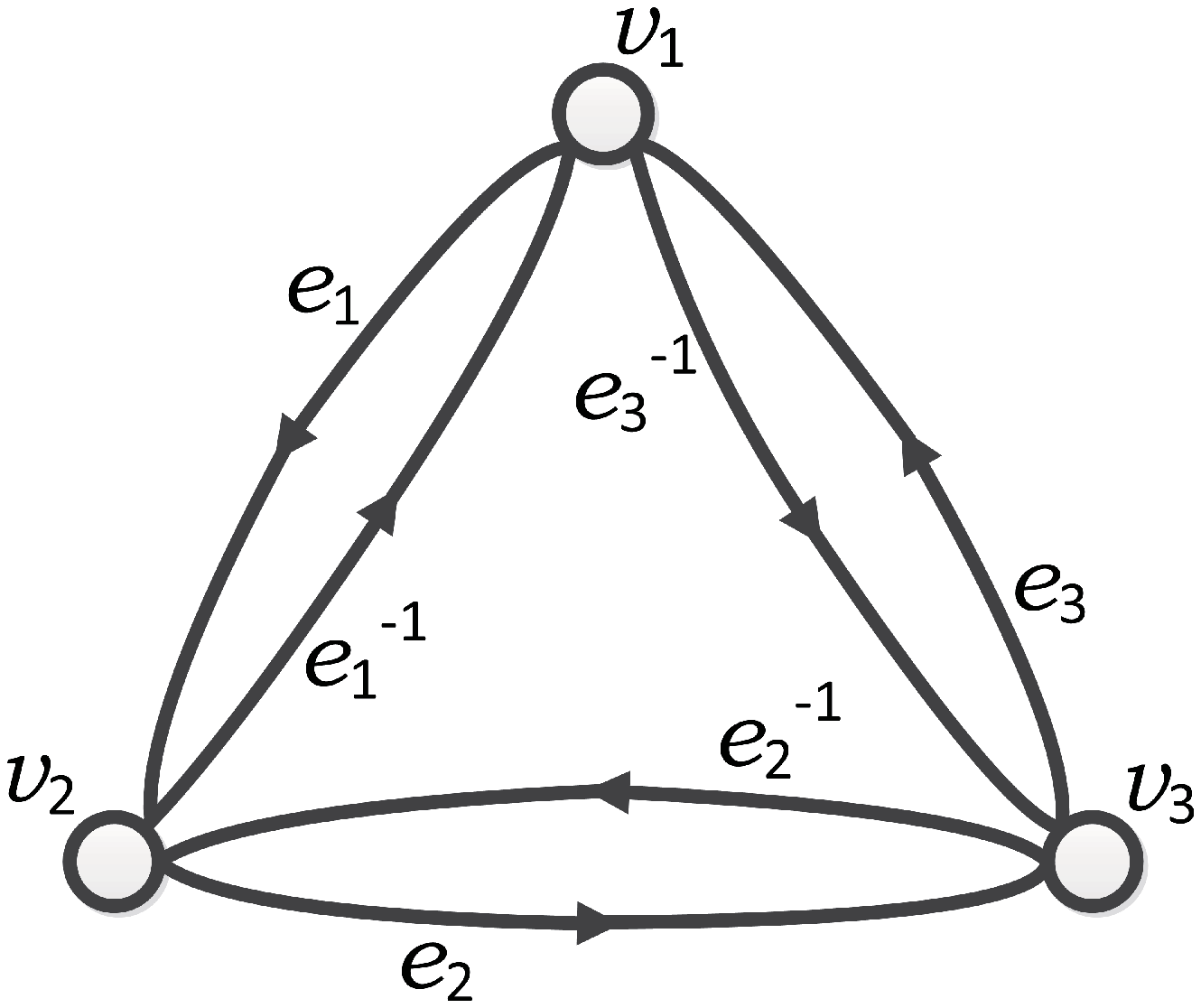}
\end{figure}

${\Spec}({\bf U}^{\rm Gro})=\Big{\{}1,1,
\frac{-1{\pm}\sqrt{3}i}{2}, \frac{-1{\pm}\sqrt{3}i}{2} \Big{\}}$.
\end{example}

\begin{example}
$\quad G=K_{1,3}$. Then $d_{o(e_1)}=d_{o(e_2)}=d_{o(e_3)}=1,\,
d_{o(e_1^{-1})}=d_{o(e_2^{-1})}=d_{o(e_3^{-1})}=3$ and ${\bf U}^{\rm Gro}$ is given as follows: 
\[
\hspace{4cm}
{\bf U}^{\rm Gro}= 
\bordermatrix{
\!\! & e_1 & e_1^{-1} & e_2 & e_2^{-1} & e_3 & e_3^{-1} \cr
e_1\!\! & 0 & 1 & 0 & 0 & 0 & 0 \cr
e_1^{-1}\!\! & -\frac{1}{3} & 0 & \frac{2}{3} & 0 & \frac{2}{3} & 0 \cr
e_2\!\! & 0 & 0 & 0 & 1 & 0 & 0 \cr
e_2^{-1}\!\! & \frac{2}{3} & 0 & -\frac{1}{3} & 0 & \frac{2}{3} & 0 \cr
e_3\!\! & 0 & 0 & 0 & 0 & 0 & 1 \cr
e_3^{-1}\!\! & \frac{2}{3} & 0 & \frac{2}{3} & 0 & -\frac{1}{3} & 0 \cr
}.
\]

\vspace{-4.0cm}
\begin{figure}[htbp] 
\hspace{2cm}
  \includegraphics[width=4.0cm]{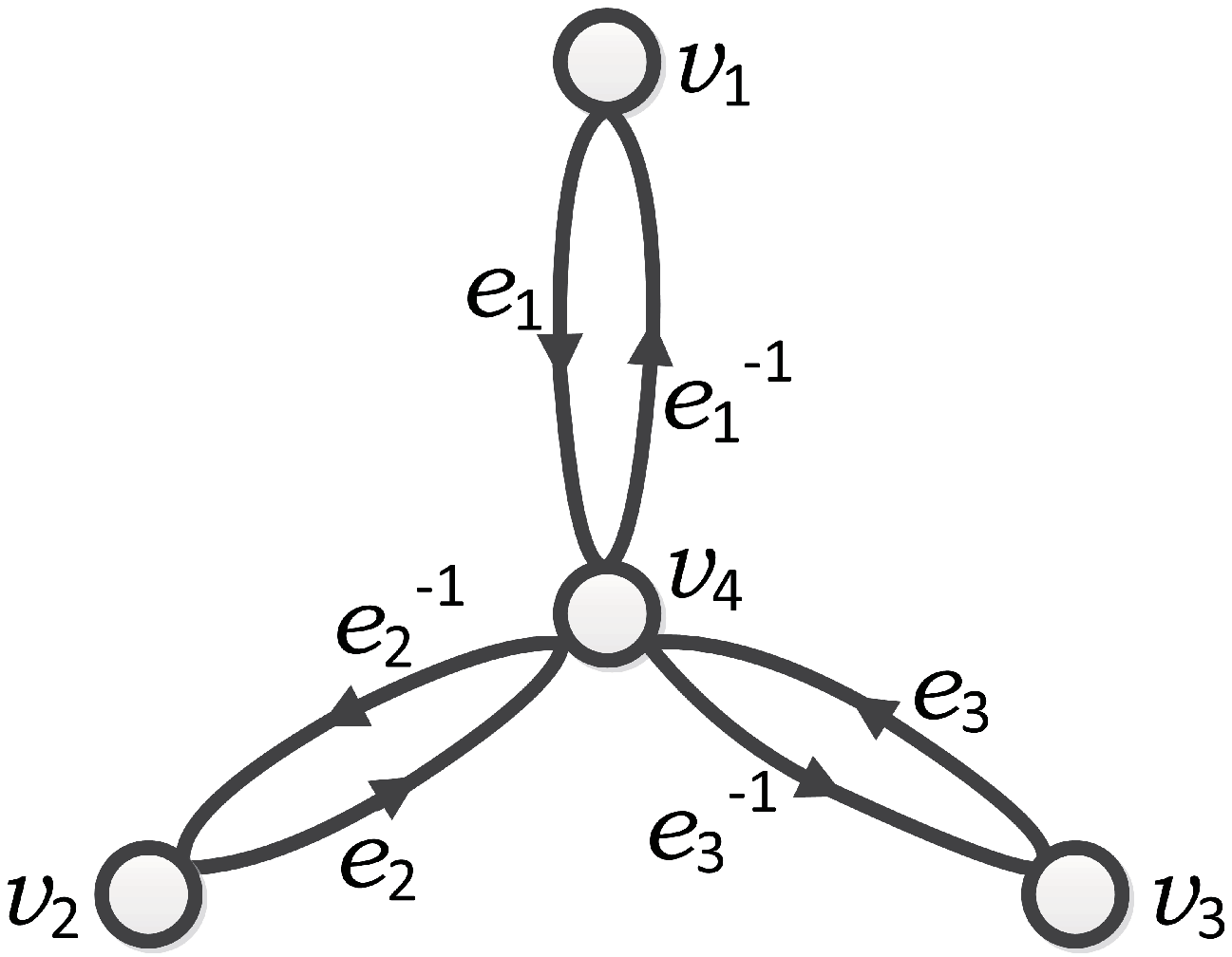}
\end{figure}

${\Spec}({\bf U}^{\rm Gro})=\Big{\{} \pm i, \pm i, 1 , -1 \Big{\}}$.
\end{example}

Let ${\bf T}=({\bf T}_{uv})_{u,v \in V(G)}$ be the $n \times n$ matrix defined as follows: 
\[
{\bf T}_{uv} =\left\{
\begin{array}{ll}
1/d_u  & \mbox{if $(u,v) \in D(G)$, } \\
0 & \mbox{otherwise.}
\end{array}
\right.
\] 
In \cite{EmmsETAL2006}, Emms et al. determined the spectrum of ${\bf U}^{\rm Gro}$ 
by using those of ${\bf T}$.

\begin{theorem}[Emms, Hancock, Severini and Wilson \cite{EmmsETAL2006}] \label{EigenGro}
Let $G$ be a connected graph with $n$ vertices and $m$ edges. 
The transition matrix ${\bf U}^{\rm Gro}$ has $2n$ eigenvalues of the form: 
\[
\lambda = \lambda {}_{\bf T} \pm i \sqrt{1- \lambda {}^2_{\bf T} } , 
\]
where $\lambda {}_{\bf T} $ is an eigenvalue of the matrix ${\bf T}$. 
The remaining $2(m-n)$ eigenvalues of ${\bf U}^{\rm Gro}$ are $\pm 1$ 
with equal multiplicities. 
\end{theorem}

As stated in Section 1, Konno and Sato \cite{KS2012} gave a concise proof of Theorem 
\ref{EigenGro} by using the second weighted zeta function of a graph.

\section{A quaternionic extension of the Grover walk on a graph}

In this section, we define a quaternionic extension of the Grover walk on $G$ and discuss 
some its properties. Beforehand, we give a brief account of quaternionic matrices and 
their right eigenvalues. 
Let $\HM$ be the set of quaternions. $\HM$ is a noncommutative associative 
algebra over $\RM$, whose underlying real vector space has dimension $4$ 
with a basis $1,i,j,k$ which satisfy the following relations: 
\[
i^2=j^2=k^2=-1,\quad ij=-ji=k,\quad jk=-kj=i,\quad ki=-ik=j.
\]
For $x=x_0+x_1i+x_2j+x_3k{\;\in\;}\HM$, $x^*=x_0-x_1i-x_2j-x_3k$ 
denotes the {\it conjugate} of $x$ in $\HM$. 
$|x|=\sqrt{xx^*}=\sqrt{x^*x}=\sqrt{x_0^2+x_1^2+x_2^2+x_3^2}$ is called the {\it norm} of $x$. 
Since $x^{-1}=x^*/|x|^2$ for a nonzero element $x{\;\in\;}\HM$, $\HM$ constitutes a skew field. 
Since quaternions do not mutually commute in general, 
we must treat left eigenvalues and right eigenvalues separately. 
In this paper, we concentrate only on right eigenvalues.

Let $\Mat(m{\times}n,\HM)$ be the set of $m{\times}n$ quaternionic matrices and 
$\Mat(n,\HM)$ the set of $n{\times}n$ quaternionic square matrices. 
For ${\bf M}{\;\in\;}\Mat(m{\times}n,\HM)$, we can write ${\bf M}={\bf M}^S+j{\bf M}^P$ 
uniquely where ${\bf M}^S,{\bf M}^P{\;\in\;}\Mat(m{\times}n,\CM)$. 
Such an expression is called the {\it symplectic decomposition} of ${\bf M}$. 
${\bf M}^S$ and ${\bf M}^P$ are called the {\it simplex part} and the {\it perplex part} 
of ${\bf M}$ respectively. 
The quaternionic conjugate ${\bf M}^*$ of a quaternionic square matrix ${\bf M}$ is obtained 
from ${\bf M}$ by taking the transpose and then taking the quaternionic conjugate of each entry. 
A quaternionic square matrix ${\bf M}$ is said to be {\it quaternionic unitary} if 
${\bf M}^*{\bf M}={\bf M}{\bf M}^*={\bf I}$.
We define $\psi$ to be the map from $\Mat(m{\times}n,\mathbb{H})$ to 
$\Mat(2m{\times}2n,\mathbb{C})$ as follows: 
\[
\psi : \Mat(m{\times}n,\mathbb{H}){\;\longrightarrow\;}\Mat(2m{\times}2n,\mathbb{C})
\quad{\bf M}{\;\mapsto\;}\begin{pmatrix}{\bf M}^S & -\overline{{\bf M}^P} \\ 
{\bf M}^P & \overline{{\bf M}^S}\end{pmatrix},
\]
where $\overline{\bf A}$ is the complex conjugate of a matrix ${\bf A}$. 
Then $\psi$ is an $\RM$-linear map. 
We can easily check that 

\begin{lemma}\label{PsiLem}
Let ${\bf M}{\;\in\;}\Mat(m{\times}n,\HM)$ and ${\bf N}{\;\in\;}
\Mat(n{\times}m,\HM)$. Then 
\[ \psi({\bf M}{\bf N})=\psi({\bf M})\psi({\bf N}). \]
\end{lemma}

Moreover, if $m=n$, then $\psi$ is an injective $\RM$-algebra homomorphism. 
We consider $\mathbb{H}^n$ as a right vector space. 
$\lambda{\;\in\;}\HM$ is said to be a right eigenvalue of ${\bf M}$ and ${\bf v}{\;\in\;}\mathbb{H}^n$ 
a right eigenvector corresponding to $\lambda$ 
if ${\bf M}{\bf v}={\bf v}{\lambda}$ for ${\bf M}{\;\in\;}\Mat(n,\mathbb{H})$. 
Now we state the facts about right eigenvalues of a quaternionic matrix as follows:

\begin{theorem}\label{QuatEigen}
For any quaternionic matrix ${\bf M}{\;\in\;}\Mat(n,\mathbb{H})$, 
there exist $2n$ complex right eigenvalues of ${\bf M}$ counted with multiplicity,
which can be obtained by solving $\det({\lambda}{\bf I}_{2n}-\psi({\bf M}))=0$. 
They appear in complex conjugate pairs $\lambda_1,\overline{\lambda_1},
{\ldots},\lambda_n,\overline{\lambda_n}$. 
The set of right eigenvalues $\sigma_r({\bf M})$ is given by 
$\sigma_r({\bf M})=\lambda_1^{\mathbb{H}^*}{\cup}{\;\cdots\;}{\cup}\lambda_n^{\mathbb{H}^*}$ 
where $\lambda^{\HM^*}=\{h^{-1}{\lambda}h\;|\;h{\;\in\;}\HM^*=\HM-\{0\}\}$. 
\end{theorem}

\begin{remark}\label{RemConjRightEigen}
If ${\bf M}{\bf v}={\bf v}{\lambda}$ then ${\bf M}{\bf v}q={\bf v}q(q^{-1}{\lambda}q)$ for every 
$q{\;\in\;}\HM^*$ and hence ${\bf v}q$ is a right eigenvector corresponding to the right eigenvalue 
$q^{-1}{\lambda}q$. 
\end{remark}

\begin{example}
${\bf M}=\begin{pmatrix}1&0\\0&i\end{pmatrix}$,\;
$\det({\lambda}{\bf I}_{4}-\psi({\bf M}))=\det\begin{pmatrix}\lambda-1&0&0&0\\0&\lambda-i&0&0\\
0&0&\lambda-1&0\\0&0&0&\lambda+i\end{pmatrix}=0$\\
${\quad \Leftrightarrow\ }\lambda=1,\,{\pm}i$,\;
$\sigma_r({\bf M})=\{1\}{\cup}\,i^{\mathbb{H}^*}$.
\end{example}

\begin{example}
${\bf M}=\begin{pmatrix}1&j\\k&i\end{pmatrix}$,\;
$\det({\lambda}{\bf I}_{4}-\psi({\bf M}))=\det\begin{pmatrix}
\lambda-1&0&0&1\\
0&\lambda-i&i&0\\
0&-1&\lambda-1&0\\
i&0&0&\lambda+i\end{pmatrix}
=0$\\
${\quad \Leftrightarrow\ }\lambda=\dfrac{1+\sqrt{3}}{2}{\pm}\dfrac{1-\sqrt{3}}{2}\,i,\;
\dfrac{1-\sqrt{3}}{2}{\pm}\dfrac{1+\sqrt{3}}{2}\,i$,\\
$\sigma_r({\bf M})=\Big{(}\dfrac{1+\sqrt{3}}{2}+\dfrac{1-\sqrt{3}}{2}\,i\Big{)}^{\mathbb{H}^*}{\cup}\,
\Big{(}\dfrac{1-\sqrt{3}}{2}+\dfrac{1+\sqrt{3}}{2}\,i\Big{)}^{\mathbb{H}^*}$.
\end{example}

Now, we give a quaternionic extension of the Grover walk on $G$. 
A discrete-time quaternionic quantum walk is a quantum process on $G$ whose 
state vector, whose entries are quaternions, 
is governed by a quaternionic unitary matrix called the quaternionic transition matrix.  
Let $G$ be a finite connected graph with $n$ vertices and $m$ edges. 
We define the state space to be the quaternionic right Hilbert space 
$\mathscr{H}_{\HM}=\oplus_{e{\in}D(G)}|e\rangle{\HM}$. 
We define the {\em quaternionic transition matrix} ${\bf U} =( {\bf U}_{ef} )_{e,f \in D(G)} $ 
of $G$ as follows:  
\begin{equation}\label{DefU}
{\bf U}_{ef} =\left\{
\begin{array}{ll}
q(e) & \mbox{if $t(f)=o(e)$ and $f \neq e^{-1} $, } \\
q(e) -1 & \mbox{if $f= e^{-1} $, } \\
0 & \mbox{otherwise,}
\end{array}
\right.
\end{equation}
where $q$ is a map from $D(G)$ to $\HM$. 
${\bf U}$ can be viewed as the time evolution operator of 
a discrete-time quaternionic quantum system.
In \cite{KMS2016}, we obtained the necessary and sufficient condition for ${\bf U}$ to be 
quaternionic unitary as follows: 

\begin{theorem}[Konno-Mitsuhashi-Sato \cite{KMS2016}]\label{UnitarityCondition}
{\ }\\
${\bf U}$ is unitary ${\Leftrightarrow}\ 
q_0(e)^2+q_1(e)^2+q_2(e)^2+q_3(e)^2-\dfrac{2q_0(e)}{d_{o(e)}}=0$, 
where $q(e)=q_0(e)+q_1(e)i+q_2(e)j+q_3(e)k$, 
and $q(e)=q(f)$ for any two arcs $e,f{\;\in\;}D(G)$ with $o(e)=o(f)$. 
\end{theorem}

From Theorem \ref{UnitarityCondition}, it follows that $q_0(e)$ must satisfy 
\begin{equation*}
0{\;\leq\;}q_0(e){\;\leq\;}\dfrac{2}{d_{o(e)}}. 
\end{equation*}
Furthermore we can readily see if $q(e)$ is positive real for each $e{\;\in\;}D(G)$, then 
${\bf U}$ must be ${\bf U}^{\rm Gro}$.

\section{The second weighted zeta function of a graph}

In this section, we give a brief summary of the second weighted zeta functions of a graph. 
We introduce an equivalence relation between cycles in $G$. 
Two cycles $C_1 =(e_1, \cdots ,e_{\ell})$ and 
$C_2 =(f_1, \cdots ,f_{\ell})$ are said to be {\em equivalent} if there exists 
$k$ such that $f_j =e_{j+k}$ for all $j$ where indices are treated modulo $\ell$. 
Let $[C]$ be the equivalence class which contains the cycle $C$. 
Let $B^r$ be the cycle obtained by going $r$ times around a cycle $B$. 
Such a cycle is called a {\em power} of $B$. 
A cycle $C$ is said to be {\em reduced} if both 
$C$ and $C^2$ has no backtracking. 
Furthermore, a cycle $C$ is said to be {\em prime} if it is not a power of 
a strictly smaller cycle.

The {\em Ihara zeta function} of a graph $G$ is 
a function of $t \in {\bf C}$ with $|t|$ sufficiently small, defined by 
\[
{\bf Z} (G, t)= {\bf Z}_G (t)= \prod_{[C]} (1- t^{ \mid C \mid } )^{-1} ,
\]
where $[C]$ runs over all equivalence classes of prime, reduced cycles 
of $G$. 

${\bf Z} (G, t)$ has two types of determinant expressions as explained below. 
Let ${\bf B}=({\bf B}_{ef})_{e,f \in D(G)} $ and 
${\bf J}_0=( {\bf J}_{ef} )_{e,f \in D(G)} $ 
be $2m{\times}2m$ matrices defined as follows: 
\begin{equation}\label{DefJ}
{\bf B}_{ef} =\left\{
\begin{array}{ll}
1 & \mbox{if $t(e)=o(f)$, } \\
0 & \mbox{otherwise,}
\end{array}
\right.
{\bf J}_{ef} =\left\{
\begin{array}{ll}
1 & \mbox{if $f= e^{-1} $, } \\
0 & \mbox{otherwise.}
\end{array}
\right.
\end{equation}
The matrix ${\bf B} - {\bf J}_0 $ is called the {\em edge matrix} of $G$. 
Then we can state two determinant expressions of ${\bf Z} (G, t)$ as follows: 

\begin{theorem}[Hashimoto \cite{Hashimoto1989}; Bass \cite{Bass1992}]
The reciprocal of the Ihara zeta function of $G$ is given by 
\begin{equation}\label{EqnDetExpressions}
{\bf Z} (G, t)^{-1} =\det ( {\bf I}_{2m} -t ( {\bf B} - {\bf J}_0 ))
=(1- t^2 )^{r-1} \det ( {\bf I}_n -t {\bf A}+ 
t^2 ({\bf D} -{\bf I}_n )), 
\end{equation}
where $r$ and ${\bf A}$ are the Betti number and the adjacency matrix 
of $G$ respectively, and ${\bf D} =({\bf D}_{uv})_{u,v{\in}V(G)}$ is the diagonal matrix 
with ${\bf D}_{uu} = \deg u $ for all $u{\;\in\;}V(G)$. 
\end{theorem}

We call the middle formula of (\ref{EqnDetExpressions}) 
the determinant expression of {\it Hashimoto type} and 
the right hand side the determinant expression of {\it Bass type} respectively.

We shall define the second weighted zeta function by using a modification 
of the edge matrix. 
Consider an $n \times n$ complex matrix 
${\bf W} =({\bf W}_{uv})_{u,v{\in}V(G)}$ with $(u,v)$-entry 
equals $0$ if $(u,v){\;\notin\;}D(G)$. 
We call ${\bf W}$ a {\em weighted matrix} of $G$.
Let $w(u,v)= {\bf W}_{uv}$ for $u,v \in V(G)$ and 
$w(e)= w(u,v)$ if $e=(u,v) \in D(G)$. Then ${\bf W}_{uv}$ is given by 
\begin{equation}\label{DefWeightedMtx}
{\bf W}_{uv} =\left\{
\begin{array}{ll}
w(e) & \mbox{if $e=(u,v){\;\in\;}D(G)$, } \\
0 & \mbox{otherwise.}
\end{array}
\right.
\end{equation}
For a weighted matrix ${\bf W}$ of $G$,  
let ${\bf B}_w=( {\bf B}^{(w)}_{ef} )_{e,f \in D(G)}$ be 
a $2m \times 2m$ complex matrix defined as follows: 
\begin{equation}\label{DefEdgeMtx}
{\bf B}^{(w)}_{ef} =\left\{
\begin{array}{ll}
w(f) & \mbox{if $t(e)=o(f)$, } \\
0 & \mbox{otherwise.}
\end{array}
\right.
\end{equation}
Then the {\em second weighted zeta function} of $G$ is defined by 
\begin{equation}\label{Def2ndWZF}
{\bf Z}_1 (G,w,t)= \det ( {\bf I}_{2m} -t ( {\bf B}_w - {\bf J}_0 ) )^{-1} . 
\end{equation}
One can consider (\ref{Def2ndWZF}) as a multi-parametrized deformation 
of the determinant expression of Hashimoto type for ${\bf Z}(G,t)$. 
If $w(e)=1$ for all $e \in D(G)$, then the second weighted zeta function of $G$ 
coincides with ${\bf Z}(G,t)$. In \cite{Sato2007}, Sato obtained the determinant expression 
of Bass type for ${\bf Z}_1 (G,w,t)$ as follows:

\begin{theorem}[Sato \cite{Sato2007}]\label{SatoThm}
\begin{equation*}
{\bf Z}_1 (G,w,t )^{-1} =(1- t^2 )^{m-n} 
\det ({\bf I}_n -t {\bf W}+ t^2 ( {\bf D}_w - {\bf I}_n )) , 
\end{equation*}
where $n=|V(G)|$, $m=|E(G)|$ and 
${\bf D}_w =({\bf D}^{(w)}_{uv})_{u,v{\in}V(G)}$ is the diagonal matrix defined by 
\begin{equation}\label{DefWeightedDegreeMtx}
{\bf D}^{(w)}_{uu} = \sum_{e:o(e)=u}w(e)
\end{equation}
for all $u{\;\in\;}V(G)$. 
\end{theorem}

\begin{remark}
We mention that taking transpose, the following equation also holds: 
\begin{equation}\label{DetExpressionTrns}
\det ( {\bf I}_{2m} -t ( {}^T\!{\bf B}_w - {\bf J}_0 ) ) =(1- t^2 )^{m-n} 
\det ({\bf I}_n -t {}^T\!{\bf W}+ t^2 ( {\bf D}_w - {\bf I}_n )), 
\end{equation}
where ${}^T\!{\bf M}$ denotes the transpose of ${\bf M}$. 
We will show a quaternionic generalization of (\ref{DetExpressionTrns}) 
and apply it to the spectral problem for our quaternionic quantum walk on $G$ 
in later sections. 
\end{remark}

\section{A quaternionic generalization of the determinant expressions}

In this section, we shall give a quaternionic generalization of (\ref{DetExpressionTrns}). 
Assume that $w(e)$ is in $\HM$ for every $e{\;\in\;}D(G)$ 
in (\ref{DefWeightedMtx}), (\ref{DefEdgeMtx}) and (\ref{DefWeightedDegreeMtx}). 
Let ${\bf K}=({\bf K}_{ev})_{e{\in}D(G),v{\in}V(G)}$ and 
${\bf L}=({\bf L}_{ev})_{e{\in}D(G),v{\in}V(G)}$ be $2m{\times}n$ matrices defined as follows: 
\begin{equation*}
{\bf K}_{ev}=
\begin{cases}
w(e) & \text{if $o(e)=v$,}\\
0 & \text{otherwise,}
\end{cases}\quad
{\bf L}_{ev}=
\begin{cases}
1 & \text{if $t(e)=v$,}\\
0 & \text{otherwise,}
\end{cases}
\end{equation*}
where column index and row index are ordered by fixed sequences 
$v_1,{\cdots},v_n$ and $e_1,{\cdots},e_{2m}$ such that $e_{2r}=e_{2r-1}^{-1}$ for $r=1,{\cdots},m$ 
respectively. 
Then we can readily see that ${}^T\!{\bf B}_w={\bf K}{}^T\!{\bf L}$. 
Let ${\bf K}={\bf K}^S+j{\bf K}^P$ be the symplectic decomposition. Then it follows that 
\begin{equation}\label{EqnPsiKandL}
\psi({\bf K})=\begin{pmatrix}
{\bf K}^S & -\overline{{\bf K}^P} \\ {\bf K}^P & \overline{{\bf K}^S}
\end{pmatrix},\quad 
\psi({\bf L})=\begin{pmatrix}
{\bf L} & 0 \\ 0 & {\bf L}
\end{pmatrix}, 
\end{equation}
and by Lemma \ref{PsiLem} we obtain 
\begin{equation*}
\psi({}^T\!{\bf B}_w)=\psi({\bf K})\psi({}^T\!{\bf L}).
\end{equation*}
In view of (\ref{EqnPsiKandL}), rows and columns of $\psi({\bf K})$ and $\psi({\bf L})$ are indexed by 
the disjoint union of two copies $D(G)_{\pm}$ of $D(G)$ and 
that of two copies $V(G)_{\pm}$ of $V(G)$ respectively. 
We denote these disjoint unions by 
\begin{equation}\label{DfnDisjUnion}
\begin{split}
D(G)_+{\dot\cup}D(G)_-&=\{{e_1}_+,{\cdots},{e_{2m}}_+,{e_1}_-,{\cdots},{e_{2m}}_-\}\quad 
({e_r}_{\pm}{\;\in\;}D(G)_{\pm}),\\
V(G)_+{\dot\cup}V(G)_-&=\{{v_1}_+,{\cdots},{v_n}_+,{v_1}_-,{\cdots},{v_n}_-\}\quad 
({v_r}_{\pm}{\;\in\;}V(G)_{\pm}),
\end{split}
\end{equation}
where ${e_r}_+$ and ${e_r}_-$ correspond to $e_r{\;\in\;}D(G)$ and 
${v_r}_+$ and ${v_r}_-$ to $v_r{\;\in\;}V(G)$. 
Orders of indices of $\psi({\bf K})$ and 
$\psi({\bf L})$ follow the alignment in (\ref{DfnDisjUnion}) so that 
${\bf B}_w$ turns out to be  
\[
\psi({}^T\!{\bf B}_w)=
\bordermatrix{
\!\! & e_{1_+}{\!\cdots\,}e_{2m_+} & \hspace{-3mm} e_{1_-}{\!\cdots\,}e_{2m_-} \cr
\substack{\vspace{-1mm}e_{1_+}\\{\vdots}\\e_{{2m}_+}} & \hspace{-3mm} \mbox{\large{${}^T\!{\bf B}_w^S$}} &  \mbox{\large{$-\overline{{}^T\!{\bf B}_w^P}$}} \cr
\substack{\vspace{-1mm}e_{1_-}\\{\vdots}\\e_{{2m}_-}} & \hspace{-3mm} \mbox{\large{${}^T\!{\bf B}_w^P$}} &  \mbox{\large{$\overline{{}^T\!{\bf B}_w^S}$}}
}
=\bordermatrix{
\!\! & e_{1_+}{\!\cdots\ }e_{2m_+} & \hspace{-3mm} e_{1_-}{\!\cdots\ }e_{2m_-} \cr
\substack{\vspace{-1mm}e_{1_+}\\{\vdots}\\e_{{2m}_+}} & \hspace{-3mm} \mbox{\large{$({\bf K}^S){}^T\!{\bf L}$}} &  \mbox{\large{$(-\overline{{\bf K}^P}){}^T\!{\bf L}$}} \cr
\substack{\vspace{-1mm}e_{1_-}\\{\vdots}\\e_{{2m}_-}} & \hspace{-3mm} \mbox{\large{$({\bf K}^P){}^T\!{\bf L}$}} &  \mbox{\large{$(\overline{{\bf K}^S}){}^T\!{\bf L}$}}
}.
\]
We also notice that ${}^T\!{\bf W}={}^T\!{\bf L}{\bf K}$. 

\begin{theorem}\label{QuatDetFormula}
Let $t$ be a complex variable. Then
\begin{equation*}
\det({\bf I}_{4m}-t\psi({}^T\!{\bf B}_w-{\bf J}_0))=
(1-t^2)^{2m-2n}
\det({\bf I}_{2n}-t\psi({}^T\!{\bf W})+t^2(\psi({\bf D}_w)-{\bf I}_{2n}))
\end{equation*}

\end{theorem}
\begin{proof}
We notice ${}^T\!{\bf B}_w={\bf K}{}^T\!{\bf L}$ at first. 
It can be shown by direct calculations that 
\begin{equation}\label{EqnSecDeterminants}
\begin{split}
\det({\bf I}_{4m}-t\psi({\bf K}{}^T\!{\bf L}-{\bf J}_0))&=
\det({\bf I}_{4m}-t\psi({\bf K})\psi({}^T\!{\bf L})+t\psi({\bf J}_0))\\
&=\det({\bf I}_{4m}-t\psi({\bf K})\psi({}^T\!{\bf L})({\bf I}_{4m}+t\psi({\bf J}_0))^{-1})
\det({\bf I}_{4m}+t\psi({\bf J}_0))\\
&=(1-t^2)^{2m}
\det({\bf I}_{4m}-t\psi({\bf K})\psi({}^T\!{\bf L})({\bf I}_{4m}+t\psi({\bf J}_0))^{-1})\\
&=(1-t^2)^{2m}
\det({\bf I}_{2n}-t\psi({}^T\!{\bf L})({\bf I}_{4m}+t\psi({\bf J}_0))^{-1}\psi({\bf K})), 
\end{split}
\end{equation}
and that 
\begin{equation*}
\psi({}^T\!{\bf L})({\bf I}_{4m}+t\psi({\bf J}_0))^{-1}\psi({\bf K})
=\begin{pmatrix}
{{}^T\!{\bf L}} & 0 \\ 0 & {{}^T\!{\bf L}}
\end{pmatrix}
\Bigg{\{}{\bf I}_{2m}{\otimes}\begin{pmatrix}
\frac{1}{1-t^2} & -\frac{t}{1-t^2} \\ -\frac{t}{1-t^2} & \frac{1}{1-t^2}\end{pmatrix}\Bigg{\}}
\begin{pmatrix}
{\bf K}^S & -\overline{{\bf K}^P} \\ {\bf K}^P & \overline{{\bf K}^S}
\end{pmatrix}.
\end{equation*}
Putting ${\bf X}=({\bf I}_{4m}+t\psi({\bf J}_0))^{-1}$, 
we see if $f'{\;\neq\;}e',{e'}^{-1}$ in $D(G)_+{\dot\cup}D(G)_-$, 
then ${\bf X}_{e'f'}=0$. 
For every $u',v'{\;\in\;}V(G)_+{\dot\cup}V(G)_-$, the $(u',v')$-entry of 
$\psi({}^T\!{\bf L}){\bf X}\psi({\bf K})$ is given by 
\begin{equation*}
(\psi({}^T\!{\bf L}){\bf X}\psi({\bf K}))_{u'v'}=\sum_{e',f'{\in}D(G)_+{\dot\cup}D(G)_-}
\psi({}^T\!{\bf L})_{u'e'}{\bf X}_{e'f'}\psi({\bf K})_{f'v'}.
\end{equation*}
Hence if $(v,u)=e{\;\in\;}D(G)$, then $f'=e'$ and it follows that 
\begin{equation*}
\begin{split}
(\psi({}^T\!{\bf L}){\bf X}\psi({\bf K}))_{u_+v_+}&=
\psi({}^T\!{\bf L})_{u_+e_+}{\bf X}_{e_+e_+}\psi({\bf K})_{e_+v_+}
=\dfrac{1}{1-t^2}w(e)^S=\dfrac{1}{1-t^2}w((v,u))^S,\\
(\psi({}^T\!{\bf L}){\bf X}\psi({\bf K}))_{u_+v_-}&=
\psi({}^T\!{\bf L})_{u_+e_+}{\bf X}_{e_+e_+}\psi({\bf K})_{e_+v_-}
=\dfrac{1}{1-t^2}(-\overline{w(e)^P})=-\dfrac{1}{1-t^2}\overline{w((v,u))^P},\\
(\psi({}^T\!{\bf L}){\bf X}\psi({\bf K}))_{u_-v_+}&=
\psi({}^T\!{\bf L})_{u_-e_-}{\bf X}_{e_-e_-}\psi({\bf K})_{e_-v_+}
=\dfrac{1}{1-t^2}w(e)^P=\dfrac{1}{1-t^2}w((v,u))^P,\\
(\psi({}^T\!{\bf L}){\bf X}\psi({\bf K}))_{u_-v_-}&=
\psi({}^T\!{\bf L})_{u_-e_-}{\bf X}_{e_-e_-}\psi({\bf K})_{e_-v_-}
=\dfrac{1}{1-t^2}\overline{w(e)^S}=\dfrac{1}{1-t^2}\overline{w((v,u))^S}.
\end{split}
\end{equation*}
Else if $u=v$, then $f'={e'}^{-1}$ and it follows that 
\begin{equation*}
\begin{split}
(\psi({}^T\!{\bf L}){\bf X}\psi({\bf K}))_{u_+u_+}&=
\sum_{\substack{e{\in}D(G)\\o(e)=u}}
\psi({}^T\!{\bf L})_{u_+e_+^{-1}}{\bf X}_{e_+^{-1}e_+}\psi({\bf K})_{e_+u_+}
=-\dfrac{t}{1-t^2}\sum_{\substack{e{\in}D(G)\\o(e)=u}}w(e)^S,\\
(\psi({}^T\!{\bf L}){\bf X}\psi({\bf K}))_{u_+u_-}&=
\sum_{\substack{e{\in}D(G)\\o(e)=u}}
\psi({}^T\!{\bf L})_{u_+e_+^{-1}}{\bf X}_{e_+^{-1}e_+}\psi({\bf K})_{e_+u_-}
=\dfrac{t}{1-t^2}\sum_{\substack{e{\in}D(G)\\o(e)=u}}\overline{w(e)^P},\\
(\psi({}^T\!{\bf L}){\bf X}\psi({\bf K}))_{u_-u_+}&=
\sum_{\substack{e{\in}D(G)\\o(e)=u}}
\psi({}^T\!{\bf L})_{u_-e_-^{-1}}{\bf X}_{e_-^{-1}e_-}\psi({\bf K})_{e_-u_+}
=-\dfrac{t}{1-t^2}\sum_{\substack{e{\in}D(G)\\o(e)=u}}w(e)^P,\\
(\psi({}^T\!{\bf L}){\bf X}\psi({\bf K}))_{u_-u_-}&=
\sum_{\substack{e{\in}D(G)\\o(e)=u}}
\psi({}^T\!{\bf L})_{u_-e_-^{-1}}{\bf X}_{e_-^{-1}e_-}\psi({\bf K})_{e_-u_-}
=-\dfrac{t}{1-t^2}\sum_{\substack{e{\in}D(G)\\o(e)=u}}\overline{w(e)^S}.
\end{split}
\end{equation*}
Otherwise, we can readily check that all of 
$(u_+,v_+),(u_+,v_-),(u_-,v_+),(u_-,v_-)$-entries of $\psi({}^T\!{\bf L}){\bf X}\psi({\bf K})$ 
equal $0$. 

Comparing entries of $\psi({}^T\!{\bf L}){\bf X}\psi({\bf K})$ with 
those of $\psi({\bf W})$ and $\psi({\bf D}_w)$, it follows that  
\begin{equation*}
\psi({}^T\!{\bf L}){\bf X}\psi({\bf K})=
\dfrac{1}{1-t^2}\psi({}^T\!{\bf W})-\dfrac{t}{1-t^2}\psi({\bf D}_w).
\end{equation*}
Consequently, we obtain the following equation from (\ref{EqnSecDeterminants}) as desired.
\begin{equation*}
\begin{split}
\det({\bf I}_{4m}-t\psi({\bf K}{}^T\!{\bf L}-{\bf J}_0))&=
(1-t^2)^{2m}
\det({\bf I}_{2n}-\dfrac{t}{1-t^2}\psi({}^T\!{\bf W})+\dfrac{t^2}{1-t^2}\psi({\bf D}_w))\\
&=(1-t^2)^{2m-2n}
\det((1-t^2){\bf I}_{2n}-t\psi({}^T\!{\bf W})+t^2\psi({\bf D}_w))\\
&=(1-t^2)^{2m-2n}
\det({\bf I}_{2n}-t\psi({}^T\!{\bf W})+t^2(\psi({\bf D}_w)-{\bf I}_{2n})).
\end{split}
\end{equation*}
\end{proof}

\section{The right spectrum of the quaternionic quantum walk on a graph}

In this section, we derive the set of right eigenvalues of the quaternionic 
transition matrix ${\bf U}$ 
by using eigenvalues of a complex matrix which can be easily derived from 
the map $q : D(G){\longrightarrow}\HM$ given in (\ref{DefU}). 
Putting $t=1/{\lambda}$ in Theorem \ref{QuatDetFormula}, we obtain 
\begin{equation}\label{EqnQuatCharactEqn}
\begin{split}
\det(\lambda{\bf I}_{4m}-\psi({}^T\!{\bf B}_w-{\bf J}_0))
&=({\lambda}^2-1)^{2m-2n}
\det({\lambda}^2{\bf I}_{2n}-{\lambda}\psi({}^T\!{\bf W})+\psi({\bf D}_w)-{\bf I}_{2n}).
\end{split}
\end{equation}
Setting $w(e)=q(e)$ and 
comparing (\ref{DefU}), (\ref{DefJ}) and (\ref{DefEdgeMtx}), we readily see 
${\bf U}={}^T\!{\bf B}_w-{\bf J}_0$. 
Thus applying (\ref{EqnQuatCharactEqn}), we obtain 
\begin{equation}\label{EqnQuatEigenEqn}
\begin{split}
\det(\lambda{\bf I}_{4m}-\psi({\bf U}))
&=({\lambda}^2-1)^{2m-2n}
\det({\lambda}^2{\bf I}_{2n}-{\lambda}\psi({}^T\!{\bf W})+\psi({\bf D}_w)-{\bf I}_{2n}).\\
\end{split}
\end{equation}
If $\psi({}^T\!{\bf W})$ and $\psi({\bf D}_w)$ are simultaneously triangularizable, 
namely, there exist a regular matrix ${\bf P}{\;\in\;}\Mat(2n,\CM)$ such that 
\begin{equation}\label{TriangWandD}
{\bf P}^{-1}\psi({}^T\!{\bf W}){\bf P}=\begin{pmatrix}
\mu_1 & 0 & \cdots & 0 \\
* & \mu_2 & \cdots & 0 \\
\vdots & \ddots & \ddots & \vdots \\
* & \cdots & * & \mu_{2n}
\end{pmatrix},\quad
{\bf P}^{-1}\psi({\bf D}_w){\bf P}=\begin{pmatrix}
\xi_1 & 0 & \cdots & 0 \\
* & \xi_2 & \cdots & 0 \\
\vdots & \ddots & \ddots & \vdots \\
* & \cdots & * & \xi_{2n}
\end{pmatrix},
\end{equation}
then by (\ref{EqnQuatEigenEqn}) the characteristic equation of $\psi({\bf U})$ 
turns out to be 
\begin{equation}\label{EqnFactQuatEigenEqn}
\begin{split}
\det(\lambda{\bf I}_{4m}-\psi({\bf U}))
&=({\lambda}^2-1)^{2m-2n}
\det({\lambda}^2{\bf I}_{2n}-{\lambda}{\bf P}^{-1}\psi({}^T\!{\bf W}){\bf P}
+{\bf P}^{-1}\psi({\bf D}_w){\bf P}-{\bf I}_{2n}).\\
&=({\lambda}^2-1)^{2m-2n}\prod_{r=1}^{2n}(\lambda^2-\lambda\mu_r+\xi_r-1)=0.
\end{split}
\end{equation}
We notice if $(\mu_r,\,\mu_s)$ ($r{\;\neq\;}s$) is a complex conjugate pair as 
stated in Theorem \ref{QuatEigen} 
then so is $(\xi_r,\,\xi_s)$. 
Observing that $\det(\lambda{\bf I}_{4m}-\psi({\bf U}))$ is a polynomial of $\lambda$, 
we obtain by solving (\ref{EqnFactQuatEigenEqn}) that 

\begin{theorem}\label{GeneralEigenFormula}
$|{\Spec}(\psi({\bf U}))|=4m$. Suppose that $\psi({}^T\!{\bf W})$ and $\psi({\bf D}_w)$ are 
simultaneously triangularizable. 
If $G$ is not a tree, then $4n$ of them are 
\begin{equation*}
\lambda = \dfrac{\mu_r \pm \sqrt{\mu_r^2-4(\xi_r-1)}}{2} \quad (r=1,\cdots,2n),
\end{equation*}
where $\mu_r{\;\in\;}{\Spec}(\psi({}^T\!{\bf W})),\xi_r{\;\in\;}{\Spec}(\psi({\bf D}_w))$ 
as presented in (\ref{TriangWandD}). 
The remaining $4(m-n)$ are $\pm 1$ with equal multiplicities. 
If $G$ is a tree, then 
\begin{equation*}
\begin{split}
&{\Spec}(\psi ({\bf U}))
=\Big{\{}\dfrac{\mu_r \pm \sqrt{\mu_r^2-4(\xi_r-1)}}{2} 
\bigl|\; r=1,\cdots,2n\Big{\}}
-\{1,1,-1,-1\}.
\end{split}
\end{equation*}
$\sigma_r({\bf U})=\bigcup_{\lambda{\in}{\Spec}(\psi({\bf U}))}\lambda^{\HM^*}$
\end{theorem}

Finally, we discuss quaternionic quantum walks which satisfy the following 
condition on $q(e)$ : 
\begin{equation}\label{QuatCond}
\displaystyle \sum_{e:o(e)=u}q(e) \text{ does not depend on $u$.} 
\end{equation}
We have investigated this case in \cite{KMS2016}. 
In short, this condition says that the sum of the entries in each column of {\bf U} 
does not depend on the column. 
Quaternionic quantum walks need not satisfy (\ref{QuatCond}) in general, 
however, the Grover walk satisfies (\ref{QuatCond}) by definition. 
Hence we can consider this condition is inherited from the Grover walk.

Since $\psi$ is injective, ${}^T\!{\bf W}{\bf D}_w={\bf D}_w{}^T\!{\bf W}$ is equivalent to 
$\psi({}^T\!{\bf W})\psi({\bf D}_w)=\psi({\bf D}_w)\psi({}^T\!{\bf W})$. 
In this case, $\psi({}^T\!{\bf W})$ and $\psi({\bf D}_w)$ are simultaneously triangularizable. 

\begin{proposition}\label{PropCommutativity}
{\rm (\ref{QuatCond})} implies ${}^T\!{\bf W}{\bf D}_w={\bf D}_w{}^T\!{\bf W}$. Moreover, 
if $w(e){\;\neq\;}0$ for every $e{\;\in\;}D(G)$, then 
{\rm (\ref{QuatCond})} ${\Leftrightarrow\ }{}^T\!{\bf W}{\bf D}_w={\bf D}_w{}^T\!{\bf W}$.
\end{proposition}

\begin{proof}
If $(v,u){\;\in\;}D(G)$, then by Theorem \ref{UnitarityCondition} we have 
\begin{equation*}
\begin{split}
({}^T\!{\bf W}{\bf D}_w)_{uv}&=w((v,u))\sum_{\substack{e{\in}D(G)\\o(e)=v}}w(e)
=d_vw((v,u))w((v,u)),\\
({\bf D}_w{}^T\!{\bf W})_{uv}&=\sum_{\substack{e{\in}D(G)\\o(e)=u}}w(e)w((v,u))
=d_uw((u,v))w((v,u)).
\end{split}
\end{equation*}
Therefore (\ref{QuatCond}) implies $d_vw((v,u))w((v,u))=d_uw((u,v))w((v,u))$ 
for all $(v,u){\;\in\;}D(G)$ and thereby 
${}^T\!{\bf W}{\bf D}_w={\bf D}_w{}^T\!{\bf W}$. 
If $w(e){\;\neq\;}0$ for every $e{\;\in\;}D(G)$, then $d_vw((v,u))w((v,u))=d_uw((u,v))w((v,u))$ 
implies $d_vw((v,u))=d_uw((u,v))$. 
Hence 
\[
\sum_{e{\in}D(G),o(e)=v}w(e)=\sum_{e{\in}D(G),o(e)=u}w(e) \text{ \ for } u,\,v{\;\in\;}V(G),
\] with 
$(v,u){\;\in\;}D(G)$. 
Since $G$ is connected, the equation just before holds for every pair of vertices 
and thereby (\ref{QuatCond}) holds. 
\end{proof}

Hence one can view 
Theorem \ref{GeneralEigenFormula} as a generalization of \cite{KMS2016}. 
By (\ref{QuatCond}), we may put $\alpha=\sum_{e:o(e)=u}q(e)$ independently of $u$. 
Then we immediately see that $q(e)={\alpha}/d_{o(e)}$ for every $e{\;\in\;}D(G)$. 
For $\alpha{\;\in\;}\mathbb{H}-\mathbb{R}$, it is known that 
there exist nonzero quaternions $h_{\pm}{\;\in\;}\HM^*$ such that 
$h_{\pm}^{-1}{\alpha}h_{\pm}=\alpha_{\pm}$ are complex numbers which are complex conjugate 
with each other. (For the details, see \cite{KMS2016})
Then we readily see 
\begin{equation*}
{\bf U}_{\pm}=h_{\pm}^{-1}{\bf U}h_{\pm}
=(({\bf U}_{\pm})_{ef})_{e,f{\in}D(G)}{\;\in\;}\Mat(2m,\CM),
\end{equation*}
where 
\begin{equation*}
({\bf U}_{\pm})_{ef} =\left\{
\begin{array}{ll}
\dfrac{\alpha_{\pm}}{d_{o(e)}} & \mbox{if $t(f)=o(e)$ and $f \neq e^{-1} $, } \\
\dfrac{\alpha_{\pm}}{d_{o(e)}} -1 & \mbox{if $f= e^{-1} $, } \\
0 & \mbox{otherwise.}
\end{array}
\right.
\end{equation*}
It follows that  
\begin{equation*}
\begin{split}
\det({\lambda}{\bf I}_{4m}-\psi({\bf U}))
&=\det(\psi(h_+{\bf I}_{2m})^{-1}({\lambda}{\bf I}_{4m}-\psi({\bf U}))\psi(h_+{\bf I}_{2m}))\\
&=\det({\lambda}{\bf I}_{4m}-\psi(h_+{\bf I}_{2m})^{-1}\psi({\bf U})\psi(h_+{\bf I}_{2m}))\\
&=\det({\lambda}{\bf I}_{4m}-\psi(h_+^{-1}{\bf U}h_+))\\
&=\det({\lambda}{\bf I}_{4m}-\psi({\bf U}_+))\\
&=\begin{vmatrix}{\lambda}{\bf I}_{2m}-{\bf U}_+&0\\0&{\lambda}{\bf I}_{2m}-{\bf U}_-\end{vmatrix}\\
&=\det({\lambda}{\bf I}_{2m}-{\bf U}_+)\det({\lambda}{\bf I}_{2m}-{\bf U}_-)\\
&=\det({\lambda}{\bf I}_{2m}-{\bf U}_+)\det({\lambda}{\bf I}_{2m}-\overline{{\bf U}_+}).
\end{split}
\end{equation*}
Therefore, we can calculate all right eigenvalues of 
${\bf U}$ by calculating eigenvalues of ${\bf U}_+$ since 
eigenvalues of ${\bf U}_-=\overline{{\bf U}_+}$ are complex conjugates of those of ${\bf U}_+$. 
Accordingly, ${\bf W}_{\pm}=h_{\pm}^{-1}{\bf W}h_{\pm}$ and 
${\bf B}_{w_{\pm}}=h_{\pm}^{-1}{\bf B}_{w}h_{\pm}
=({\bf B}^{(w_{\pm})})_{ef}$ are given by 
\begin{equation*}
({\bf W}_{\pm})_{uv} =\left\{
\begin{array}{ll}
\dfrac{{\alpha}_{\pm}}{d_u} & \mbox{if $(u,v){\;\in\;}D(G)$, } \\
0 & \mbox{otherwise,}
\end{array}
\right.
\quad 
{\bf B}^{(w_{\pm})}_{ef} =\left\{
\begin{array}{ll}
\dfrac{{\alpha}_{\pm}}{d_{o(f)}} & \mbox{if $t(e)=o(f)$, } \\
0 & \mbox{otherwise.}
\end{array}
\right.
\end{equation*}
In this case, we can apply (\ref{DetExpressionTrns}) to obtain 
the next theorem which is a special case of Theorem \ref{GeneralEigenFormula}. 

\begin{theorem}[Konno-Mitsuhashi-Sato \cite{KMS2016}]\label{EigenFormula}
$|{\Spec}(\psi({\bf U}))|=4m$. Suppose that (\ref{QuatCond}) holds. 
If $G$ is not a tree, then $4n$ of them are 
\begin{equation*}
\lambda = \dfrac{\mu_+ \pm \sqrt{\mu_+^2-4(\alpha_+-1)}}{2},\; 
\dfrac{\mu_- \pm \sqrt{\mu_-^2-4(\alpha_--1)}}{2},
\end{equation*}
where $\mu_{\pm}{\;\in\;}{\Spec}({}^T\!{\bf W}_{\pm})$. 
The remaining $4(m-n)$ are $\pm 1$ with equal multiplicities. 
If $G$ is a tree, then 
\begin{equation*}
\begin{split}
&{\Spec}(\psi ({\bf U}))\\
&=\Big{\{}\dfrac{\mu_+ \pm \sqrt{\mu_+^2-4(\alpha_+-1)}}{2},\;
\dfrac{\mu_- \pm \sqrt{\mu_-^2-4(\alpha_--1)}}{2} 
\bigl|\; \mu_{\pm} {\in} {\Spec}({}^T\!{\bf W}_{\pm})\Big{\}}
-\{1,1,-1,-1\}.
\end{split}
\end{equation*}
$\sigma_r({\bf U})=\bigcup_{\lambda{\in}{\Spec}({\bf U}_+)}\lambda^{\HM^*}$
\end{theorem}

We will show an example which does not satisfy (\ref{QuatCond}). 

\newpage
\begin{example}
$G=K_{1,3}$. Let $w(e_1)=1+i,\,w(e_2)=1-j,\,w(e_3)=2,\,
w(e_1^{-1})=w(e_2^{-1})=w(e_3^{-1})=0$. 
\begin{figure}[htbp] 
  \includegraphics[width=4.0cm]{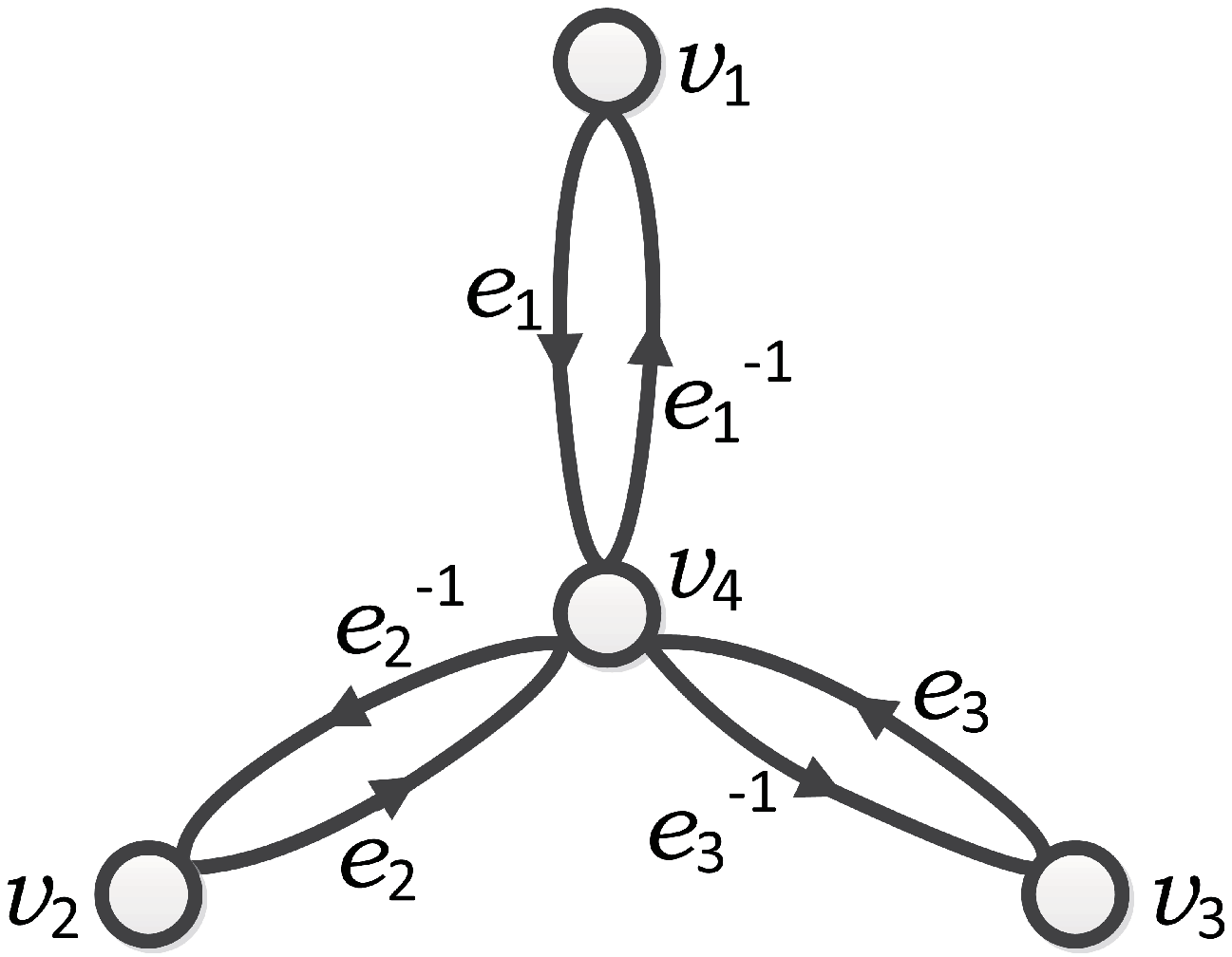}
\end{figure}
\vspace{-5.0cm}
\begin{equation*}
\begin{split}
\hspace{4cm}
&{}^T\!{\bf W}=\bordermatrix{
 &v_1&v_2&v_3&v_4\cr
v_1&0&0&0&0 \cr
v_2&0&0&0&0 \cr
v_3&0&0&0&0 \cr
v_4&1+i&1-j&2&0
}, \\
&\psi({}^T\!{\bf W})=\begin{pmatrix}
0&0&0&0&0&0&0&0 \\
0&0&0&0&0&0&0&0 \\
0&0&0&0&0&0&0&0 \\
1+i&1&2&0&0&1&0&0 \\
0&0&0&0&0&0&0&0 \\
0&0&0&0&0&0&0&0 \\
0&0&0&0&0&0&0&0 \\
0&-1&0&0&1-i&1&2&0
\end{pmatrix}.
\end{split}
\end{equation*}
Then ${\bf U}$, ${\bf D}_w$ and $\psi({\bf D}_w)$ are given by 
\begin{equation*}
\begin{split}
&{\bf U}=\begin{pmatrix}
0&i&0&0&0&0\\
-1&0&0&0&0&0\\
0&0&0&-j&0&0\\
0&0&-1&0&0&0\\
0&0&0&0&0&1\\
0&0&0&0&-1&0\\
\end{pmatrix},
{\bf D}_w=\begin{pmatrix}
1+i&0&0&0 \\
0&1-j&0&0 \\
0&0&2&0 \\
0&0&0&0
\end{pmatrix},\\
&\psi({\bf D}_w)=\begin{pmatrix}
1+i&0&0&0&0&0&0&0 \\
0&1&0&0&0&1&0&0 \\
0&0&2&0&0&0&0&0 \\
0&0&0&0&0&0&0&0 \\
0&0&0&0&1-i&0&0&0 \\
0&-1&0&0&0&1&0&0 \\
0&0&0&0&0&0&2&0 \\
0&0&0&0&0&0&0&0
\end{pmatrix}.
\end{split}
\end{equation*}
Let ${\bf P}$ be defined by 
\[
{\bf P}=\begin{pmatrix}
1&0&0&0&0&0&0&0\\
0&1&1&0&0&0&0&0\\
0&0&0&0&1&0&0&0\\
0&0&0&0&0&0&1&0\\
0&0&0&1&0&0&0&0\\
0&i&-i&0&0&0&0&0\\
0&0&0&0&0&1&0&0\\
0&0&0&0&0&0&0&1
\end{pmatrix}.
\]
Then
\begin{equation*}
\begin{split}
{\bf P}^{-1}\psi({}^T\!{\bf W}){\bf P}&=
\begin{pmatrix}
0&0&0&0&0&0&0&0\\
0&0&0&0&0&0&0&0\\
0&0&0&0&0&0&0&0\\
0&0&0&0&0&0&0&0\\
0&0&0&0&0&0&0&0\\
0&0&0&0&0&0&0&0\\
1+i&1+i&1-i&0&2&0&0&0\\
0&-1+i&-1-i&1-i&0&2&0&0
\end{pmatrix},\\
{\bf P}^{-1}\psi({\bf D}_w){\bf P}&=
\begin{pmatrix}
1+i&0&0&0&0&0&0&0\\
0&1+i&0&0&0&0&0&0\\
0&0&1-i&0&0&0&0&0\\
0&0&0&1-i&0&0&0&0\\
0&0&0&0&2&0&0&0\\
0&0&0&0&0&2&0&0\\
0&0&0&0&0&0&0&0\\
0&0&0&0&0&0&0&0
\end{pmatrix}.
\end{split}
\end{equation*}
Hence it follows that 
$\Spec(\psi({}^T\!{\bf W}))=\{\mu_1,\mu_2,{\cdots},\mu_8\}
=\{0,0,0,0,0,0,0,0\}$, 
$\Spec(\psi({\bf D}_w))=\{\xi_1,\xi_2,{\cdots},\xi_8\}
=\{1+i,1+i,1-i,1-i,2,2,0,0\}$. 
Applying Theorem \ref{GeneralEigenFormula}, we obtain 
\begin{equation*}
\begin{split}
\Spec(\psi({\bf U}))&=\{{\pm}\sqrt{-i},{\pm}\sqrt{-i},{\pm}\sqrt{i},{\pm}\sqrt{i},
{\pm}\sqrt{-1},{\pm}\sqrt{-1},{\pm}1,{\pm}1\}-\{1,1,-1,-1\}\\
&=\{{\pm}\dfrac{1-i}{\sqrt{2}},{\pm}\dfrac{1-i}{\sqrt{2}},
{\pm}\dfrac{1+i}{\sqrt{2}},{\pm}\dfrac{1+i}{\sqrt{2}},
{\pm}i,{\pm}i\}.
\end{split}
\end{equation*}
Since $-i=j^{-1}ij{\;\in\;}i^{\HM^*}$, an eigenvalue and its complex conjugate belong to 
the same set of all quaternionic conjugations of the eigenvalue. 
Thus $\sigma_r({\bf U})=i^{\HM^*}{\cup}\Big{(}\dfrac{1+i}{\sqrt{2}}\Big{)}^{\!\HM^*}{\cup}
\Big{(}-\dfrac{1+i}{\sqrt{2}}\Big{)}^{\!\HM^*}$. 
\end{example}

\section*{Acknowledgments}

The first author is partially supported by the Grant-in-Aid for Scientific Research 
(Challenging Exploratory Research) of Japan Society for the Promotion of Science (Grant No. 15K13443).
The second author is partially supported by the Grant-in-Aid for Scientific Research 
(C) of Japan Society for the Promotion of Science (Grant No. 16K05249). 
The third author is partially supported by the Grant-in-Aid for Scientific Research 
(C) of Japan Society for the Promotion of Science (Grant No. 15K04985). 
We are grateful to K. Tamano, S. Matsutani and Y. Ide for some valuable comments on this work. 
%
%

\end{document}